\documentclass[10pt,letterpaper]{article}

\usepackage[hang,perpage]{footmisc}

\usepackage{amsmath}
\usepackage{amsthm}
\usepackage{amssymb}
\usepackage{stmaryrd}
\usepackage{graphicx}
\usepackage{latexsym}
\usepackage{times}
\usepackage[usenames]{color}
\usepackage{braket}
\usepackage{mathabx}
\usepackage{textcomp}
\usepackage{makeidx}
\usepackage{tocbibind}
\usepackage[bookmarks=true,backref=true,colorlinks=true,linkcolor=darkolivegreen,
citecolor=darkolivegreen,urlcolor=darkolivegreen]{hyperref}
\numberwithin{equation}{subsection}
\usepackage{authblk}

\theoremstyle{definition}
\newtheorem{ass}{Assumption}[section]
\newtheorem{theorem}[ass]{Theorem}
\newtheorem{lemma}[ass]{Lemma}
\newtheorem{definition}[ass]{Definition}
\newtheorem{corollary}[ass]{Corollary}

\def\indexname{Index of terminology}
\makeindex

\newcommand{\captionfonts}{\footnotesize}
\makeatletter  
\long\def\@makecaption#1#2{%
  \vskip\abovecaptionskip
  \sbox\@tempboxa{{\captionfonts #1: #2}}%
  \ifdim \wd\@tempboxa >\hsize
    {\captionfonts #1: #2\par}
  \else
    \hbox to\hsize{\hfil\box\@tempboxa\hfil}%
  \fi
  \vskip\belowcaptionskip}
\makeatother   
\definecolor{darkolivegreen}{rgb}{0.333333, 0.419608, 0.1843140}
\allowdisplaybreaks

\makeatletter
\def\printnotation{{%
\def\indexname{Index of notation}
\begin{theindex}
\@input{\jobname.ntn}
\end{theindex}
}}
\makeatother

\makeglossary

\begin{document}


\title{Stability study of a model for the Klein-Gordon equation  
in Kerr space-time}

\author[1,2,3]{Horst Reinhard Beyer}
\author[3]{Miguel Alcubierre}
\author[3]{Miguel Megevand}
\author[3,4]{Juan Carlos Degollado}
\affil[1]{Instituto Tecnol\'ogico Superior de Uruapan, 
Carr. Uruapan-Carapan No. 5555, Col. La Basilia, Uruapan, 
Michoac\'an. M\'exico}
\affil[2]{Theoretical Astrophysics, IAAT, Eberhard Karls 
University of T\"ubingen, T\"ubingen 72076, Germany}
\affil[3]{Instituto de Ciencias Nucleares, Universidad Nacional
  Aut\'onoma de M\'exico, Circuito Exterior C.U., A.P. 70-543,
  M\'exico D.F. 04510, M\'exico}
\affil[4]{ Instituto de Astronom\'{\i}a, Universidad Nacional
  Aut\'onoma de M\'exico, Circuito Exterior C.U., A.P. 70-264,
  M\'exico D.F. 04510, M\'exico}
\date{\today}                                     

\maketitle

\begin{abstract}

The current early stage in the investigation of the stability 
of the Kerr metric is characterized by the study of appropriate 
model problems. Particularly interesting is the problem 
of the stability of the solutions of the Klein-Gordon equation, 
describing the propagation of a scalar field of mass $\mu$
in the background of a rotating black hole. Rigorous 
results proof the stability of the reduced, by separation 
in the azimuth angle in 
Boyer-Lindquist coordinates, field for sufficiently large masses.
Some, but not all, numerical investigations find  
instability of the reduced field for rotational parameters $a$
extremely close to $1$. Among others, the paper 
derives a model problem for the equation which supports 
the instability of the field down to 
$a/M \approx 0.97$.  
 
\end{abstract}

\section{Introduction}
\label{introduction}

The discussion of the stability of the Kerr black hole is
still in its early stages. The first intermediate goal is the proof or disproof 
of its stability under ``small'' perturbations.
\newline
\linebreak  
In the case of the Schwarzschild metric, by using the 
Regge-Wheeler-Zerilli-Moncrief (RWZM) decomposition of fields 
in a Schwarzschild background \cite{regge,zerilli,moncrief,teukolsky,
press,kegeles}, the question of 
the stability of linearized fields can be 
completely reduced to the question of the stability of the 
solutions 
of the wave equation on Schwarzschild space-time. In comparison 
to Schwarzschild space-time, 
the case of Kerr space-time is complicated by a lower dimensional
symmetry group and the absence of a Killing field 
that is everywhere
time-like outside the horizon. For instance, the latter is reflected 
in the fact that energy densities corresponding to the Klein-Gordon field
in a Kerr gravitational field have no definite sign. 
This absence complicates the application
of methods from operator theory, as used in this paper, and also 
of so called  
``energy methods'' that are both employed in estimating
the decay of solutions of hyperbolic partial differential 
equations.
\newline
\linebreak
For  
Kerr space-time, a reduction similar to RWZM for Schwarzschild space-time
is not known. On the other hand, the 
finding of a symmetry operator, containing only time 
derivatives up to first order, for a rescaled wave 
operator gives hope that such a reduction might 
exist \cite{beyercraciun}. 
If such 
reduction exists,  
there is no guarantee that the relevant equation is the 
scalar wave equation. It is quite possible that such 
equation contains 
an additional (even positive) potential term that, similar to the 
potential term introduced by a mass of the field, see below, 
could result in instability of the solutions. Second, an
instability of a massive scalar field in a Kerr background 
could indicate instability of the metric against perturbations 
by matter which generically has mass. If this were the case, 
even a proof of the stability 
of Kerr space-time could turn out as 
a purely mathematical exercise with little
relevance for general relativity. 
\newline
\linebreak
Currently, the 
main focus is the study of  
the stability of the solutions of the Klein-Gordon field on 
a Kerr background with the hope that the results lead to insight 
into the problem of linearized stability. Although the results 
of this paper also apply to the case that $\mu = 0$, its 
main focus is the case of Klein-Gordon fields 
of mass $\mu > 0$.
\newline
\linebreak
Quite differently from the case of a Schwarzschild background, the
results for these test cases suggest an
asymmetry between the cases 
$\mu = 0$ and $\mu \neq 0$. In the case of 
the wave equation, i.e., $\mu = 0$, rigorous mathematical, analytical and numerical  
results point to the stability of the solutions 
\cite{whiting,finster,dafermos,andersson,krivan1,krivan2}.
On the other hand, for $\mu \neq 0$, there are a number of analytical 
and numerical results 
pointing in the direction of instability of the solutions 
under certain conditions \cite{damour,detweiler,zouros,furuhashi,khanna,
cardoso,hod}.
\newline
\linebreak
In particular,  
unstable modes were found by the  
numerical investigations by 
Furuhashi and Nambu for $\mu M \sim 1$ and $(a/M) = 0.98$,
by Strafuss and Khanna for $\mu M \sim 1$ 
and $(a/M) = 0.9999$ and by Cardoso and Yoshida for 
$\mu M \leqslant 1$ and $0.98 \leqslant (a/M) < 1$. The 
analytical study by Hod and Hod finds
unstable modes for $\mu M \sim 1$ with a growth rate
which is four orders of 
magnitude larger than previous estimates. On the other hand, 
\cite{beyer1} proves
that the restrictions of 
the solutions of the separated, in the azimuthal coordinate, Klein-Gordon 
field (RKG) are stable for 
\begin{equation} \label{massineq}
\mu \geqslant \frac{|m|a}{2Mr_{+}} 
\sqrt{1 + \frac{2M}{r_{+}} + \frac{a^2}{r_{+}^2}}
\, \, .
\end{equation}
Here
$m \in {\mathbb{Z}}$ is the 
`azimuthal separation parameter' and 
\begin{equation*}
r_{+} := M + 
\sqrt{M^2 - a^2} \, \, .
\end{equation*}
The paper \cite{beyer4}, among others, gives a stronger estimate.
It proves that 
the solutions of the RKG are stable for $\mu$
satisfying 
\begin{equation*} 
\mu \geqslant \frac{|m|a}{2Mr_{+}} 
\sqrt{1 + \frac{2M}{r_{+}}}
\, \, .
\end{equation*}
So far, these have been the only mathematically 
rigorous results on the stability of the solutions of the 
RKG for $\mu > 0$. The first estimate has been confirmed 
numerically by Konoplya and Zhidenko \cite{konoplya}. 
The stronger estimate has been confirmed by S. Hod \cite{hod1}. 
These results 
contradict the result of Zouros and Eardley, 
but are 
consistent with 
the other results above. In addition, the numerical
result by Konoplya and Zhidenko {\it finds no unstable modes 
of the RKG for $\mu M \ll 1$ and $\mu M \sim 1$}, contradicting
above cited analytical and numerical investigations.
\newline
\linebreak
The situation for the RKG for the case 
$\mu > 0$ is quite puzzling. The RKG originates from the 
case $\mu = 0$ by 
the addition of a positive bounded potential term
to the equation. Without the presence of a first order time 
derivative in the equation,
from this alone, 
by help of methods from operator theory, it would be easy 
to prove that the stability of the solutions of the wave equation
implies the stability of the solutions of the 
Klein-Gordon equation for non-vanishing mass. 
Also, the energy estimates in Lemma~$4.7$ of \cite{beyer4}, 
indicate a stabilizing influence of such a term.
On the other hand, so far, there is no result that 
would allow to draw such conclusion. 
The numerical 
results that indicate instability in the case $\mu \neq 0$ 
make quite special assumptions on the values of the rotational 
parameter of the black hole that could indicate numerical 
artifacts. 
Moreover, as mentioned before, the numerical investigation 
by Konoplya 
et al. does not find any unstable 
modes and contradicts all these investigations. Also 
the analytical results in this area are not accompanied by 
error estimates and therefore ultimately 
inconclusive.
\newline
\linebreak
The present paper tries to help clearing up this 
situation. It derives and analyzes the stability 
of a simplified model equation 
from the Klein-Gordon equation, describing 
the propagation of a complex scalar field  
of mass $\mu \geqslant 0$ 
in the gravitational field of a rotating Kerr black hole.
Also, the simplified model equation is in certain
sense ``mathematically controllable.''
The latter term will be given 
meaning in the next section.
\newline
\linebreak
The remainder of the paper is organized as follows.
Section~$2$ gives the derivation of the simplified model
equation.  
Section~$3$ gives 
basic properties of operators read off 
from this approximate equation which 
provide 
the basis for the formulation of an initial-value 
problem along with a stability analysis for that equation 
given in the same section. 
The results of Section~$3$ are used in 
Section~$4$ to show the instability of the solutions 
of the approximate equation for moderately large 
rotation parameters $a$.   
Finally, the paper concludes 
with a discussion of the results and an appendix 
containing the proof of a result that is used  
in the main text. 

\section{Derivation of the Approximate Equation}

The goal of this section is the description of the 
relationship between the Klein-Gordon equation
and the approximate equation. Of course, this relationship 
is crucial for the model character of the latter
equation. In particular, there are various ways to arrive at 
approximate equations suitable for analytical 
treatment. 
Therefore in the following, we describe the derivation 
of the approximate equation in more detail. We mention that 
it is easily seen that our approach is in many aspects different 
from that 
of Detweiler in \cite{detweiler}.
\newline 
\linebreak   
To facilitate analytical treatment, the spectral problem,
that is associated with the approximate equation, is 
required to lead on a 
special case of the Coulomb wave equation. 
The latter can be reduced to an equation of 
confluent hypergeometric type. The solutions of
these equations are well-known. 
\newline
\linebreak
We start from the Klein-Gordon equation, describing 
the propagation of a complex scalar field 
$u \in C^2(\Omega,{\mathbb{C}})$ 
of mass $\mu \geqslant 0$ 
in the gravitational field of a rotating Kerr black hole of 
mass $M > 0$ and with rotational parameter $a \in [0,M]$
in the following form 
\begin{align} \label{kleingordon} 
& \left[1 - \frac{a^2 \triangle}{(r^2+a^2)^2} \, 
\sin^2\theta \right]
\frac{\partial^2 u}{\partial t^2} +  
\frac{4Mar}{(r^2+a^2)^2}\frac{\partial^2 u}{\partial t \partial \varphi} + 
\\
& 
\frac{\triangle}{(r^2+a^2)^2} \cdot 
\left[ - \frac{\partial}{\partial r} \triangle  \frac{\partial}{\partial r}
 - \frac{1}{\sin \theta} \,
\frac{\partial}{\partial \theta} \sin \theta \,\frac{\partial}{\partial \theta}
- \left(
\frac{1}{\sin^2 \theta} - \frac{a^2}{\triangle} \right) \frac{\partial^2}{\partial \varphi^2} 
+ \mu^2 \Sigma \right] u = 0 \, \, . \nonumber
\end{align}
Here\footnote{If not 
otherwise indicated, the symbols 
$t,r,\theta,\varphi$ denote coordinate projections whose 
domains will be obvious from the context. In addition, 
we assume 
the composition 
of maps, which includes addition, multiplication and so forth, 
always to be maximally defined. For instance, the sum of two 
complex-valued maps is defined on the intersection of their domains.
Finally, we use 
Planck units where 
the reduced Planck constant $\hbar$, the 
speed of light in vacuum $c$, 
and the gravitational constant $G$, all have the numerical 
value $1$.},
$(t,r,\theta,\varphi) : 
\Omega \rightarrow {\mathbb{R}}^4$ are the 
Boyer-Lindquist coordinates and
\begin{align*} 
& \Delta := r^2 - 2Mr +a^2 
\, \, , \, \, \Sigma := r^2 + a^2 
\cos^2\!\theta \, \, , \\
& r_{\pm} := M \pm \sqrt{M^2-a^2}
\, \, , \, \, \Omega := {\mathbb{R}} \times (r_{+},\infty) \times
(-\pi,\pi) \times (0,\pi) \, \, .
\end{align*}
The introduction of the new unknown $\bar{u} \in C^2(\Omega,{\mathbb{C}})$ by 
\begin{equation} \label{substitution}
\bar{u} := \triangle^{1/2} u
\end{equation}
and subsequent multiplication of the resulting equation by 
$(r^2+a^2)^2/\triangle^2$ gives
\begin{align*}
& - \frac{\partial^2\bar{u}}{\partial r^2 } + \left\{ 
\left[1 - \frac{a^2 \triangle}{(r^2+a^2)^2} \, 
\sin^2\theta \right] \frac{\partial^2}{\partial t^2}
+ \mu^2
\right\} \bar{u} \\
& + \frac{1}{\triangle} 
\left\{4M(r+M) \left[1 - \frac{a^2 \triangle}{(r^2+a^2)^2} \, 
\sin^2\theta \right] \frac{\partial^2}{\partial t^2} \right. \\
& \left. 
- \frac{1}{\sin \theta} \,
\frac{\partial}{\partial \theta} \sin \theta \,\frac{\partial}{\partial \theta}
 - \frac{1}{\sin^2 \theta} \, \frac{\partial^2}{\partial \varphi^2}  
+ \mu^2 (2Mr - a^2 \sin^2\theta)
\right\} \bar{u} \\
& + \frac{1}{\triangle^2} \left\{ 
4M^2(2Mr-a^2) \left[1 - \frac{a^2 \triangle}{(r^2+a^2)^2} \, 
\sin^2\theta \right] \frac{\partial^2}{\partial t^2} \right. \\
& \left. + 4 M a r \,
\frac{\partial^2 u}{\partial t \partial \varphi} + a^2 \, 
\frac{\partial^2 u}{\partial \varphi^2} - (M^2 - a^2)
\right\} \bar{u} = 0 \, \, .
\end{align*}
In a first step towards the model equation, we neglect the 
term 
\begin{equation*}
\- \frac{a^2 \triangle}{(r^2+a^2)^2} \, 
\sin^2\theta \, \, ,
\end{equation*}
inside the 
factors 
\begin{equation*}
\left[1 - \frac{a^2 \triangle}{(r^2+a^2)^2} \, 
\sin^2\theta \right] \, \, ,
\end{equation*}
multiplying the second order time derivatives of $\bar{u}$. 
In addition, we neglect the term 
\begin{equation*}
-a^2 \sin^2\theta
\end{equation*} 
of the factor
\begin{equation*} 
2Mr - a^2 \sin^2\theta
\end{equation*}
multiplying the square of the mass of the field, $\mu^2$.
Through this, we arrive at an equation whose
angular derivatives and potential 
leave spherical harmonics 
invariant.  
The second step replaces the factor 
\begin{equation*}
\triangle = (r-r_{-})(r-r_{+}) \, \, , 
\end{equation*} 
inside the factors multiplying terms enclosed by curly brackets,
by 
\begin{equation*}
(R-r_{-})(r-r_{+}) \, \, , 
\end{equation*}
where $R \in [r_{+},\infty)$ is a new 
parameter. Finally, the third step replaces the $r$ inside
terms enclosed by curly brackets by the constant map 
of value $R$. Through the last two 
steps, we arrive, after separation of the 
$t$, $\theta$, and $\varphi$
variables,
at a special case of the Coulomb wave equation. 
The latter can be reduced to an equation of 
confluent hypergeometric type. The solutions of
these equations are well-known. Through 
the previous steps,
we arrive at a form of the model equation
\begin{align} \label{modelequation1}
& 0 = - \frac{\partial^2\bar{u}}{\partial r^2} + \left( 
\frac{\partial^2}{\partial t^2}
+ \mu^2
\right) \bar{u} \\
&
+ \frac{1}{R_{-}(r-r_{+})} 
\left[4M(R+M) \, \frac{\partial^2}{\partial t^2} 
- \frac{1}{\sin \theta} \,
\frac{\partial}{\partial \theta} \sin \theta \,
\frac{\partial}{\partial \theta}
 - \frac{1}{\sin^2 \theta} \, \frac{\partial^2}{\partial \varphi^2}  
+ 2 M R \mu^2 
\right] \bar{u} \nonumber \\
& + \frac{1}{R_{-}^2(r-r_{+})^2} \left[
4M^2(2MR - a^2) \, \frac{\partial^2}{\partial t^2} 
+ 4 M a R \,
\frac{\partial^2 u}{\partial t \partial \varphi} + a^2 \, 
\frac{\partial^2 u}{\partial \varphi^2} - (M^2 - a^2)
\right] \bar{u} \, \, , \nonumber
\end{align}
where 
\begin{equation*}
R_{-} := R - r_{-} \, \, .
\end{equation*}
We note that $R_{-} = 0$ if and only if $R=r_{+}$
and $a = M$. For this reason, in the following, we implicitly 
assume for the case $R=r_{+}$ that $a \in [0,M)$.  
Substitution of (\ref{substitution}) into the (\ref{modelequation1})
leads to the model 
equation for (\ref{kleingordon})
\begin{align} \label{modelequation2}
& \left[ 
1 + \frac{4M(R+M)}{R_{-}(r-r_{+})} + 
\frac{4M^2(2MR-a^2)}{R_{-}^2(r-r_{+})^2}
\right] \frac{\partial^2 u}{\partial t^2} + 
\frac{4MaR}{R_{-}^2(r-r_{+})^2}\, 
\frac{\partial^2 u}{\partial t \partial \varphi} \nonumber \\
& + \left\{ 
- \frac{1}{\triangle} \, \frac{\partial}{\partial r} 
\triangle \frac{\partial}{\partial r} + 
\frac{1}{R_{-}(r-r_{+})} \left[ - \frac{1}{\sin \theta} \,
\frac{\partial}{\partial \theta} \sin \theta \,
\frac{\partial}{\partial \theta}
 - \frac{1}{\sin^2 \theta} \, \frac{\partial^2}{\partial \varphi^2}  
+ 2 M R \mu^2 \right] \nonumber
\right. \\
& \left. + \frac{1}{R_{-}^2(r-r_{+})^2} \left[ a^2 
\frac{\partial^2}{\partial \varphi^2} + (M^2-a^2)
\left(\frac{R_{-}^2}{(r-r_{-})^2} - 1 \right) \right] + \mu^2 
\right\} u = 0 \, \, .
\end{align}
We note that 
\begin{align*}
& M(R+M) \geq  M(r_{+}+M) = \frac{1}{2} \, \left( 
r_{+}^2 + a^2 + 2 M^2
\right) > 0 \, \, , \\
& 2M R - a^2 \geq 2 M r_{+} - a^2 = r_{+}^2 
> 0 \, \, .
\end{align*} 
Division of (\ref{modelequation2}) by 
\begin{align*}
& \rho := 
1 + \frac{4M(R+M)}{R_{-}(r-r_{+})} + 
\frac{4M^2(2MR-a^2)}{R_{-}^2(r-r_{+})^2} \\
& = 
\frac{R_{-}^2(r-r_{+})^2 
+ 4M(R+M) R_{-}(r-r_{+}) + 4M^2(2MR-a^2)}{R_{-}^2(r-r_{+})^2} \\
& = \frac{\left[R_{-}(r-r_{+})+2M(R+M)\right]^2
-4M^2(R^2+M^2+a^2)}{R_{-}^2(r-r_{+})^2} > 0
\end{align*}
gives
\begin{align} \label{modelequation3}
&\frac{\partial^2 u}{\partial t^2} + 
\frac{4MaR}{\rho R_{-}^2(r-r_{+})^2}\, 
\frac{\partial^2 u}{\partial t \partial \varphi} \nonumber \\
& + \frac{1}{\rho} \left\{ 
- \frac{1}{\triangle} \, \frac{\partial}{\partial r} 
\triangle \frac{\partial}{\partial r} + 
\frac{1}{R_{-}(r-r_{+})} \left[ - \frac{1}{\sin \theta} \,
\frac{\partial}{\partial \theta} \sin \theta \,
\frac{\partial}{\partial \theta}
 - \frac{1}{\sin^2 \theta} \, \frac{\partial^2}{\partial \varphi^2}  
+ 2 M R \mu^2 \right] \nonumber
\right. \\
& \left. + \frac{1}{R_{-}^2(r-r_{+})^2} \left[ a^2 
\frac{\partial^2}{\partial \varphi^2} + (M^2-a^2)
\left(\frac{R_{-}^2}{(r-r_{-})^2} - 1 \right) \right] + \mu^2 
\right\} u = 0 \, \, .
\end{align}

Substitution of the ansatz  
\begin{equation*}
u(t,r,\theta,\varphi) = v(t,r) Y_{lm}(\theta,\varphi) \, \, , \, \, 
\end{equation*}
where $v \in C^2({\mathbb{R}} \times I,{\mathbb{C}})$, $I := (r_{+},\infty)$, 
$m \in {\mathbb{Z}}$ and $l \in 
\{ |m|,|m|+1,\dots\}$, $c_l := l(l+1)$ 
and $Y_{lm}$
the spherical harmonics corresponding to $m$ and $l$, 
into (\ref{modelequation3}) gives the reduced equation

\begin{align} \label{modelequation4}
&\frac{\partial^2 v}{\partial t^2} + 
i b \,
\frac{\partial v}{\partial t} +
\left(
- \frac{1}{\rho \triangle} \, \frac{\partial}{\partial r} 
\triangle \frac{\partial}{\partial r} + V \right) v = 0 
\, \, . 
\end{align}
Here 
\begin{align*} 
& b := 
\frac{4maMR}{\rho R_{-}^2(r-r_{+})^2} 
= \frac{4maMR}{R_{-}^2(r-r_{+})^2 
+ 4M(R+M) R_{-}(r-r_{+}) + 4M^2(2MR-a^2)} \nonumber \\
& V := \frac{1}{\rho} \left\{\mu^2 +
\frac{c_l  
+ 2 M R \mu^2 }{R_{-}(r-r_{+})} 
- \frac{
(M^2-a^2)
\left[1 - \frac{R_{-}^2}{(r-r_{-})^2}\right] +
 m^2 a^2}{R_{-}^2(r-r_{+})^2} 
\right\} = \\
& \frac{\mu^2 
R_{-}^2(r-r_{+})^2+(c_{l}  
+ 2 M R \mu^2)R_{-}(r-r_{+}) - \left\{
(M^2-a^2)
\left[1 - \frac{R_{-}^2}{(r-r_{-})^2}\right] + m^2 a^2 
\right\}}{R_{-}^2(r-r_{+})^2 
+ 4M(R+M) R_{-}(r-r_{+}) + 4M^2(2MR-a^2)}  
\, \, . \nonumber
\end{align*}
We note in particular that $b$ and
$V$ are bounded continuous 
real-valued functions that have continuous extensions to 
$\bar{I}$.
We define 
\begin{equation*} 
V_0 := \inf_{r \in I} V 
\end{equation*}
and note that 
\begin{equation*}
V_0 \geq 
- \, \frac{m^2 a^2 + M^2 -a^2}{4M^2(2MR-a^2)} \, \, .
\end{equation*}

\begin{figure} 
\centering
\includegraphics[width=5.6cm,height=5.6cm]{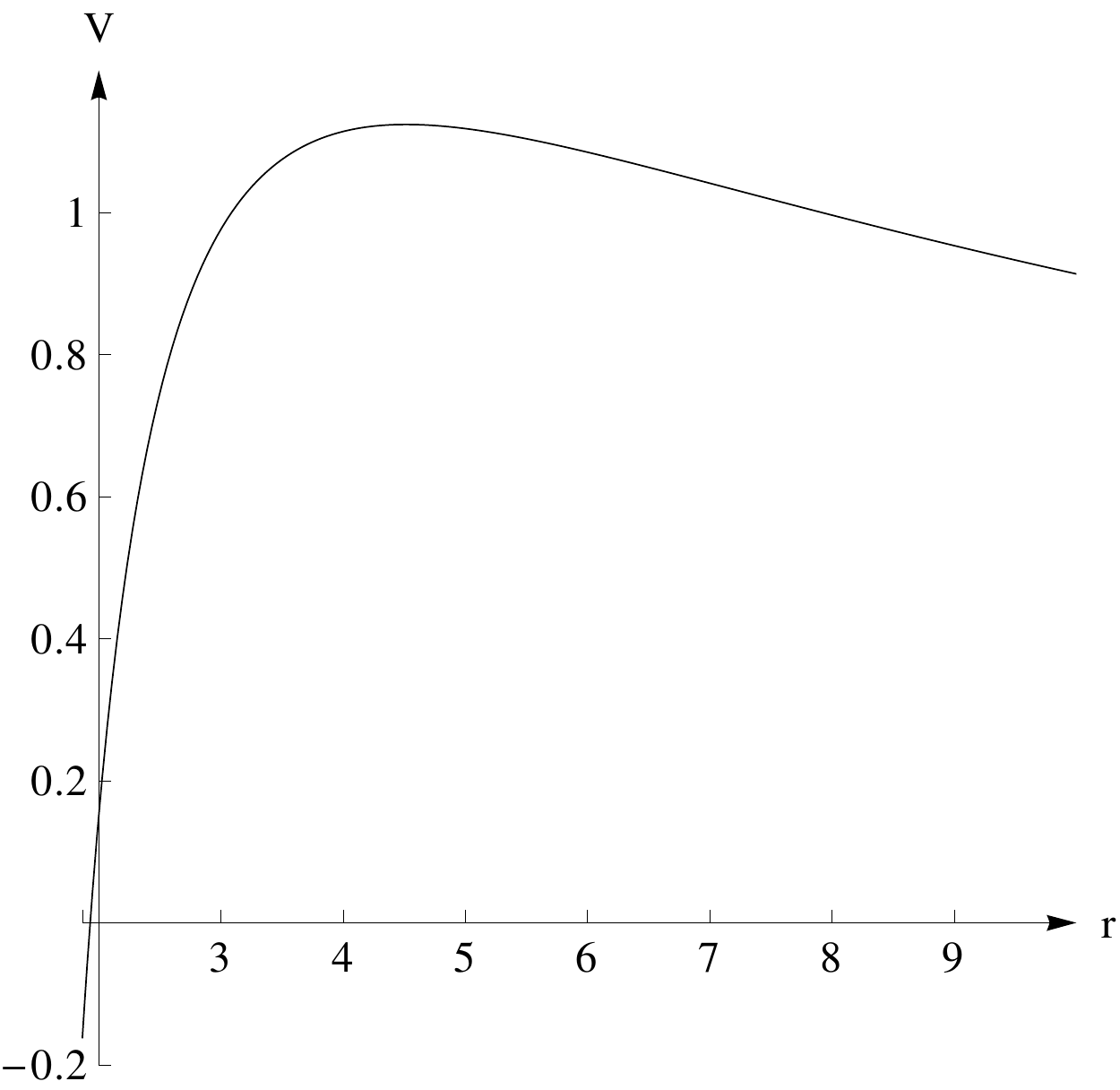}
\caption{Sketch of $\textrm{Graph}(V)$ from 
$r_{+}$ to $10$,
for $M=1$, 
$a=1/2$, $m=3$, $R = r_{+}$, $\mu = 0.5$ and $l = 20$.}
\label{figV}
\end{figure}

\section{Formulation of an Initial Value Problem
for the Reduced Equation}

This section gives 
basic properties of operators read off 
from the approximate equation (\ref{modelequation4}) 
as well as  
the formulation of an initial-value 
problem along with a stability analysis for the latter, 
following \cite{beyer,beyer2}.
\newline
\linebreak 
As data space for (\ref{modelequation4}), we choose the 
weighted $L^2$-space
\begin{equation*}
X := L^2_{\mathbb{C}}(I,\rho \triangle) \, \, . 
\end{equation*}
Since 
\begin{equation*}
\rho \, \triangle = \triangle + \frac{4M(R+M)}{R_{-}} \, (r-r_{-}) + 
\frac{4M^2(2MR-a^2)(r-r_{-})}{R_{-}^2(r-r_{+})} \, \, ,
\end{equation*}
this implies, in particular, that all continuous functions 
that have continuous extensions to $\bar{I}$ which are 
non-vanishing on $r_{+}$ are not part of the data space.

\begin{definition}
We define $B$ as maximal multiplication operator in $X$ 
corresponding to $b$.
\end{definition}

\begin{lemma}
$B$ is a bounded linear 
and self-adjoint operator on $X$. $B$ is positive 
for $m \geq 0$ and negative for $m \leq 0$.
\end{lemma} 

\begin{proof}
Since $b$ is in 
particular  bounded an real-valued, $B$ is a bounded linear 
and self-adjoint operator on $X$. Further, since $b$ is positive, negative for $m \geq 0$ and $m \leq 0$, respectively, 
$B$ is positive, negative for $m \geq 0$ and $m \leq 0$, 
respectively.
\end{proof}

\begin{definition}
We define the linear operator $A_0$ in 
$X$ by 
\begin{equation*}
A_0 f :=
- \frac{1}{\rho \triangle} \, (\triangle f^{\, \prime})^{\, \prime} + V f 
\end{equation*}
for every $f \in C^2_0(I,{\mathbb{R}})$.
\end{definition}
 
\begin{lemma} \label{propertiesofA0}
$A_0$ is a densely-defined, linear, symmetric, 
semibounded and 
essentially self-adjoint operator in $X$ with lower 
bound $V_0$.
\end{lemma}

\begin{proof}
First, we note that by 
\begin{equation*}
U f := (\rho \triangle)^{-1/2} f
\end{equation*}
for every $f \in L^2_{\mathbb{C}}(I)$, there is defined 
a Hilbert space isomorphism 
\begin{equation*}
U : L^2_{\mathbb{C}}(I) \rightarrow
L^2_{\mathbb{C}}(I,\rho \triangle) \, \, . 
\end{equation*}
In particular, the operator 
\begin{equation*}
A_{0U} := U ^{-1} A_0 U
\end{equation*} 
in 
$L^2_{\mathbb{C}}(I)$
has the domain  
$C^2_0(I,{\mathbb{R}})$ and is given by 
\begin{equation*}
A_{0U} f =  
- (f^{\, \prime}/\rho)^{\, \prime} + 
(V + V_U) f 
\end{equation*}
for every $f \in C^2_0(I,{\mathbb{R}})$, where
\begin{align*}
& V_U := 
- \frac{1}{4} \, ( \triangle^{\prime \, 2} + 2 \triangle 
\triangle^{\prime \prime}) F^{-1} + \frac{\triangle}{2}
\, (2 \triangle^{\prime} F^{\, \prime} + \triangle 
F^{\, \prime \prime}) F^{-2} - \frac{3}{4}\, \triangle^2
 F^{\, \prime \, 2} F^{-3} \, \, , \\
& F := \rho \triangle^2 \, \, . 
\end{align*}
We note that $V_U$ is in particular bounded. 
Hence it follows also that 
$V + V_U$ is bounded. 
We define the auxiliary operator $A_{0UT}$ in 
$L^2_{\mathbb{C}}(I)$ by 
\begin{equation*}
A_{0UT} f := - (f^{\, \prime}/\rho)^{\, \prime}
\end{equation*}
for every $f \in C^2_0(I,{\mathbb{C}})$. 
Since $C^2_0(I,{\mathbb{C}})$ is dense in 
$L^2_{\mathbb{C}}(I)$, 
$A_{0UT}$ and $A_{0U}$ are densely-defined. Also, 
$A_{0UT}$ and $A_{0U}$ are 
obviously linear. Further by use of integration 
by parts along with the boundedness of $V + V_U$, it 
follows that $A_{0UT}$, $A_{0U}$ are symmetric and 
semibounded from below. The equation
\begin{equation*} 
- (f^{\, \prime}/\rho)^{\, \prime}=0
\end{equation*}
admits the $C^2$-solution 
\begin{equation*}
f := 
r - r_{+}  + \frac{4M(R+M)}{R_{-}} \ln(r-r_{+}) - 
\frac{4M^2(2MR-a^2)}{R_{-}^2(r-r_{+})}  \, \, .
\end{equation*}
Since $f$ is not in $L^2_{\mathbb{C}}(I)$
at both ends of $I$, it follows that $A_{0UT}$ is in the 
limit point case at $r_{+}$ and at $+\infty$. Hence 
$A_{0UT}$
is essentially self-adjoint (see, e.g., \cite{weidmann}). 
Further, 
since $V + V_U$ is bounded and real-valued, 
it follows from the Rellich-Kato theorem, 
e.g., see Theorem X.12 
in Volume II of \cite{reedsimon}, that $A_{0U}$ is 
also essentially 
self-adjoint and that the domains of the closures 
$\bar{A}_{0UT}$
of $A_{0UT}$ and $\bar{A}_{0U}$ of 
$A_{0U}$ coincide. Hence, we conclude that 
$A_0$ is a densely-defined, linear, symmetric, 
semibounded and 
essentially self-adjoint operator in $X$. Finally, it follows
from integration by parts that 
$V_0$ is a lower bound for $A_0$.
\end{proof}

\begin{definition}
In the following, we set $A := \bar{A}_0$. 
\end{definition}

From Lemma~\ref{propertiesofA0}, it follows that  
\begin{theorem}
$A$ is a densely-defined, linear, symmetric, 
semibounded and self-adjoint operator in $X$ with lower 
bound $V_0$.
\end{theorem}

Further, from the 
proof of Lemma~\ref{propertiesofA0}, 
we can conclude the following, see, e.g., \cite{weidmann}.

\begin{corollary} \label{morepropertiesofA0}
\begin{itemize}
\item[]
\item[(i)]
For every $z \in 
{\mathbb{C}} \setminus ({\mathbb{R}} \times \{0\})$,
$f \in L^2_{\mathbb{C}}(I,\rho \triangle)$
and $r \in I$, 
\begin{align*}
[(A - z)^{-1}f](r)  =  
& - v_2(r) \int_{r_{+}}^{r} v_1(r^{\,\prime}) 
(\rho \triangle f)(r^{\,\prime}) 
dr^{\,\prime} \\
& - 
v_1(r) \int_{r}^{\infty} v_2(r^{\,\prime}) 
(\rho \triangle f)(r^{\,\prime}) dr^{\,\prime} \, \, , 
\end{align*}
where
$v_1,v_2$ are existing $C^2$-functions satisfying 
\begin{align*}
& - \frac{1}{\rho \triangle} \, (\triangle 
v_1^{\, \prime})^{\, \prime} + (V-z) v_1 =  
- \frac{1}{\rho \triangle} \, (\triangle 
v_2^{\, \prime})^{\, \prime} + (V - z) v_2  = 0 \, \, , \\
& v_1|_{(r_{+},c]} \in L^2_{\mathbb{C}}((r_{+},c],\rho 
\triangle)
\, \, , \, \,  v_2|_{[c,\infty)} 
\in L^2_{\mathbb{C}}([c,\infty),\rho \triangle) \, \, , \\
& [\triangle(v_1 v_2^{\, \prime} - 
v_1^{\, \prime} v_2)](c) = 1  \, \, .
\end{align*}    
for some $c \in I$.
\item[(ii)] By 
\begin{equation*}
A_{0*}v := - \frac{1}{\rho \triangle} \, (\triangle 
v^{\, \prime})^{\, \prime} + V v 
\end{equation*}
for every 
\begin{align*}
v \in & \, \{f \in C^2(I,{\mathbb{C}}) \cap L^2_{\mathbb{C}}(I,\rho 
\triangle): - \frac{1}{\rho \triangle} \, (\triangle 
f^{\, \prime})^{\, \prime} + V f \in L^2_{\mathbb{C}}(I,\rho 
\triangle)\} \\
= & \, \{f \in C^2(I,{\mathbb{C}}) \cap L^2_{\mathbb{C}}(I,\rho 
\triangle): - \frac{1}{\rho \triangle} \, (\triangle 
f^{\, \prime})^{\, \prime} \in L^2_{\mathbb{C}}(I,\rho 
\triangle)\} \, \, , 
\end{align*}
there is defined a linear 
extension of $A_0$, whose closure coincides
with $A$.
\end{itemize}
\end{corollary}

As a consequence of the self-adjointness of $A$ 
and semiboundedness of 
$A$ with lower bound $V_0$,
the objects $X, A_{-V_0+\varepsilon} := A 
- V_0 + \varepsilon,B$ and 
$C : - (- V_0 + \varepsilon)$ are easily seen 
to satisfy Assumptions~$1$ 
and $4$ of \cite{beyer}.\footnote{See also the Section~$5.1$ 
on `Damped wave equations' in \cite{beyer2}.} Here 
$\varepsilon > 0$
is assumed to have the dimension $l^{-2}$. The exact value 
of 
$\varepsilon$ does not affect subsequent results in any 
essential way. 
Application of the results of \cite{beyer} 
give, in particular, the following well-posed formulation of the 
initial value problem for (\ref{modelequation4}).

\begin{theorem} ({\bf Formulation of an initial value problem
for the reduced equation})
\begin{itemize}
\item[(i)] By
\begin{equation*} 
Y := D(A_{-V_0+\varepsilon}^{1/2}) \times X 
\end{equation*}
and 
\begin{equation*}
(\xi | \eta) := \braket{A_{-V_0+\varepsilon}^{1/2}\,\xi_1|
A_{-V_0+\varepsilon}^{1/2}\,\eta_1} + \braket{\xi_2|\eta_2}
\end{equation*}
for all $ \xi = (\xi_1,\xi_2), \eta = (\eta_1,\eta_2) \in Y$,
there is defined a complex Hilbert space $(Y,(\, |\, ))$. 
\item[(ii)] The operators
$G$ and $-G$ defined by 
\begin{equation*}
G(\xi,\eta) := (-\eta, A \, \xi + iB\eta) \, \, 
\end{equation*}
for all $\xi \in D(A)$ and $\eta \in 
D(A_{-V_0+\varepsilon}^{1/2})$
are infinitesimal generators of strongly continuous semigroups
$T_{+} : [0,\infty) \rightarrow L(Y,Y)$ and 
$T_{-} : [0,\infty) \rightarrow L(Y,Y)$, respectively. 
\item[(iii)]
For every $t_0 \in {\mathbb{R}}$ and every $\xi \in  
D(A) \times D(A_{-V_0+\varepsilon}^{1/2})$, 
there is a uniquely determined 
differentiable map $u : {\mathbb{R}} \rightarrow Y$
such that 
\begin{equation*}
u(t_0) = \xi
\end{equation*}
and 
\begin{equation} \label{eveq}
u^{\prime}(t) = - G u(t) 
\end{equation}
for all $t \in {\mathbb{R}}$. Here $\, ^{\prime}$ denotes differentiation 
of functions assuming values in $Y$. Moreover, this
$u$ is given by 
\begin{equation} \label{rep}
u(t) :=     
 \left\{
 \begin{array}{cl}
 T_{+}(t)\xi & \text{for $t \geqslant 0$} \\
 T_{-}(-t)\xi & \text{for $t < 0$}
 \end{array}
 \right.  
\end{equation}
for all $t \in {\mathbb{R}}$.
\item[(iv)]
For all $t \in [0,\infty)$:
\begin{equation*} 
|T_{\pm}(t)| \leqslant \exp(\|C\| \, 
\|A_{-V_0+\varepsilon}^{-1/2}\|t) \, \, \, ,
\end{equation*} 
where $|\,|$, $\| \, \|$ denote the operator norm for $L(Y,Y)$ and
$L(X,X)$, respectively.
\item[(v)] 
$\lambda \in {\mathbb{C}}$ is a spectral value, eigenvalue
of $iG$ if and only if
\begin{equation*}
A - \lambda B - \lambda^2 
\end{equation*}
is not bijective and not injective, respectively.
\item[(vi)] For any $\lambda$ from the resolvent set of $iG$ and any 
$\eta = (\eta_1,\eta_2) \in Y$ one has:
\begin{equation*}
(iG - \lambda)^{-1} \eta = \left(\xi,i(\lambda \xi + \eta_1)\right)
\, \, ,
\end{equation*} 
where 
\begin{equation*}
\xi := (A - \lambda B - \lambda^2)^{-1} \left[ (B+\lambda)\eta_1 
- i \eta_2\right] \, \, . 
\end{equation*}
\end{itemize}
\end{theorem}

Equation (\ref{eveq}) is the interpretation of (\ref{modelequation4})
used in this paper. In this sense, (iv) 
shows the well-posedness of
the initial value problem for (\ref{modelequation4}), i.e., the 
existence and uniqueness of the solution and its continuous dependence 
on the initial data. Moreover (\ref{rep}) gives a representation
of the solution and (iii) gives a rough bound for its growth
in time. In general, this bound is not strong enough to show, 
if applicable,  
stability
of the solutions to (\ref{modelequation4}). Part (v) reduces the 
determination of the generator's spectrum
to the determination of the spectrum
of the operator polynomial $A - \lambda B - \lambda^2 \, , \, \lambda \in 
{\mathbb{C}}$ \cite{markus,rodman}. Moreover (vi) does 
the same for the resolvents.

\begin{lemma} \label{eigenfunctionsI}
For every $\lambda \in {\mathbb{C}}$,
\begin{equation*} 
\ker(A - \lambda B - \lambda^2) = 
\ker(A_{0*} - \lambda B - \lambda^2) \, \, .
\end{equation*}
\end{lemma}

\begin{proof}
Let $\lambda \in {\mathbb{C}}$ and $f \in 
\ker(A - \lambda B - \lambda^2)$. Then 
\begin{equation*}
f = (A + i)^{-1}(\lambda^2 + \lambda B + i) f \, \, .
\end{equation*}
Hence the statement follows from 
Corollary~\ref{morepropertiesofA0} and 
the fact that $b$ is a $C^{\infty}$-function.
\end{proof}

\begin{ass} In the following, we use the conventions that 
$\sqrt{z}$, for some 
$z \in {\mathbb{C}}$, denotes some 
square root of $z$, whereas $[\, \,]^{1/2} : {\mathbb{C}} \setminus 
(-\infty,0] \rightarrow {\mathbb{C}}$ 
denotes the principal part of the complex square root function, 
\end{ass}

\begin{theorem} \label{t:eigenfunction}
If $\lambda \in {\mathbb{C}} \setminus 
\{-\mu,\mu\}$, $\gamma \in 
{\mathbb{C}}$ satisfies 
\begin{equation} \label{eqgamma}
\gamma^2 - \gamma + \frac{1}{R_{-}^2} \, [
4 M^2 (2 M R - a^2) \lambda^2 + 4 m a M R \lambda + 
m^2 a^2 + M^2 - a^2] = 0 \, \, , 
\end{equation} 
then 
the $C^2$-solutions of 
\begin{equation} \label{eigenfunctions}
- \frac{1}{\rho \triangle} \, (\triangle 
v^{\, \prime})^{\, \prime} + (V - \lambda b - \lambda^2) v 
= 0
\end{equation}
are given by 
\begin{equation*} 
\triangle^{-1/2} \, (r - r_{+})^{\gamma} \, 
e^{\sqrt{\mu^2 - \lambda^2} \, r} 
w(-2 \sqrt{\mu^2 - \lambda^2} \, (r-r_{+})) \, \, , 
\end{equation*}
where $w : \Omega \rightarrow {\mathbb{C}}$ is some
solution of Kummer's  
equation in the complex domain
\begin{equation*}  
z w^{\, \prime \prime}(z) + (2 \gamma - z) w^{\, \prime}(z) -
\alpha w(z) \, \, , 
\end{equation*}
$z \in U$, 
where $U \subset {\mathbb{C}}$ is an open 
neighborhood of the half-ray 
\begin{equation*}
- \sqrt{\mu^2 - \lambda^2} \, . \, (0,\infty)
\end{equation*}
and 
\begin{equation*}
\alpha := \gamma + \frac{
4 M (R + M) \lambda^2 - c_l 
- 2M R \mu^2}{2R_{-}\sqrt{\mu^2 - \lambda^2}} \, \, .
\end{equation*} 
\end{theorem}

\begin{proof}
Direct calculation. 
\end{proof}

Hence, it follows from known properties 
of the solutions of the confluent hypergeometric 
equation that 

\begin{corollary} \label{eigenfunctionsII}
If $\lambda \in {\mathbb{C}} \setminus 
\{-\mu,\mu\}$, $\gamma \in 
{\mathbb{C}}$ satisfies (\ref{eqgamma}), 
and $2 \gamma \notin {\mathbb{Z}}$, then all $C^2$-solutions 
of (\ref{eigenfunctions}) are linear combinations of 
the linear independent solutions
\begin{align*}
& \triangle^{-1/2} \, (r - r_{+})^{\gamma} \, 
e^{\sqrt{\mu^2 - \lambda^2} \, r} 
M(\alpha,2 \gamma, 
-2 \sqrt{\mu^2 - \lambda^2} \, (r-r_{+})) \, \, , \\
& \triangle^{-1/2} \, (r - r_{+})^{1 - \gamma} \, 
e^{\sqrt{\mu^2 - \lambda^2} \, r} 
M(1 + \alpha - 2 \gamma,2(1 - \gamma), 
-2 \sqrt{\mu^2 - \lambda^2} \, (r-r_{+})) \, \, , 
\end{align*}
where for $\beta \in {\mathbb{C}} \setminus 
{\mathbb{Z}}$, $M(\alpha,\beta,\cdot)$
denotes the confluent hypergeometric function 
of the first kind, as defined in \cite{abramowitz}.
\end{corollary}

For $\lambda \in {\mathbb{C}}$, 
(\ref{eqgamma}) has the solutions 

\begin{equation} \label{gammas}
\gamma = \frac{1}{2} \pm \frac{1}{R_{-}}
\, \sqrt{- 4 M^2 (2 M R - a^2) \lambda^2 - 
4 m a M R \lambda - 
m^2 a^2 - (M^2 - a^2) + \frac{R_{-}^2}{4}} \, \, .
\end{equation}
In the special case 
$R = r_{+}$, the latter reduces
to 
\begin{equation} \label{gammas1}
\gamma = \frac{1}{2} \pm 
\frac{i}{2 (M^2 - a^2)^{1/2}}(2 M r_{+} \lambda + ma) \, \, .
\end{equation} 
Hence in this special case, the condition 
$2 \gamma \notin {\mathbb{Z}}$ is equivalent
to the condition   
\begin{equation} \label{exceptions1}
\lambda \neq 
- \, \frac{1}{2 M r_{+}} \, [ma + i k (M^2 - a^2)^{1/2}]
\, \, , 
\end{equation}
where $k \in {\mathbb{Z}}$.
Further, straightforward calculation shows that argument 
of the square root in (\ref{gammas}) 
is in $(-\infty,0]$ if and only 
if $\lambda \in S$, where  
\begin{align*}
& S := 
 \bigg(- \infty, \frac{- m a R 
- \left[
m^2 a^2 \triangle(R) + (2MR-a^2)|M^2 - a^2 - (R_{-}^2/4)|
\right]^{1/2} }{2M(2MR-a^2)} \bigg] \\
&   
\cup \bigg[\frac{ - m a R +
\left[ 
m^2 a^2 \triangle(R) + (2MR-a^2)|M^2 - a^2 - (R_{-}^2/4)|
\right]^{1/2} }{2M(2MR-a^2)}, \infty\bigg) 
\, \, (\subset {\mathbb{R}})\, \, 
.
\end{align*}
In particular, we note for the special case 
$R = r_{+}$ that 
\begin{equation*}
S = {\mathbb{R}} \, \, . 
\end{equation*}

Hence
\begin{lemma}
For $\lambda \in {\mathbb{C}}$,
(\ref{eqgamma}) has a solution with real part $>1/2$
if and only if $\lambda \in {\mathbb{C}} \setminus S$. In the 
remaining cases, this equation has only solutions 
with real part $1/2$. If $\lambda \in 
{\mathbb{C}} \setminus S$, the solution 
with real part $>1/2$ is given by  
\begin{equation} \label{L2gamma}
\frac{1}{2} + \frac{1}{R_{-}}
\left[- 4 M^2 (2 M R - a^2) \lambda^2 - 
4 m a M R \lambda - 
m^2 a^2 - (M^2 - a^2) + \frac{R_{-}^2}{4} \right]^{1/2} 
\, \, ,
\end{equation}
and the other solution of (\ref{eqgamma}) has 
real part $<1/2$.   
\end{lemma}

Hence, we conclude from 
Corollary~\ref{eigenfunctionsII} that 

\begin{corollary}
If $\lambda \in {\mathbb{C}} \setminus 
\{-\mu,\mu\}$, and the solutions of 
(\ref{eqgamma}) are not part of 
$(1/2).{\mathbb{Z}}$, then there is 
a non-trivial solution $v$ to (\ref{eigenfunctions}) such that 
$\chi_{_{(r_{+},c)}} \cdot v \in X$ for some $c > r_{+}$ 
if and only if  
$\lambda \in {\mathbb{C}} \setminus S$. In the latter case, 
all such $v$ are multiples of 
\begin{equation*}
\triangle^{-1/2} \, (r - r_{+})^{\gamma} \, 
e^{\sqrt{\mu^2 - \lambda^2} \, r} 
M(\alpha,2 \gamma, 
-2 \sqrt{\mu^2 - \lambda^2} \, (r-r_{+})) \, \, ,
\end{equation*} 
where $\gamma$ is given by (\ref{L2gamma}).
\end{corollary}

If $\lambda \in {\mathbb{C}} \setminus 
((-\infty,-\mu] \cup [\mu,\infty))$, and the solutions of 
(\ref{eqgamma}) are not part of $(1/2).{\mathbb{Z}}$,
it follows from the asymptotic 
properties of the confluent hypergeometric function 
of the first kind, see, e.g., 
\cite{abramowitz}~13.1.4, 
that 
\begin{equation*}
\triangle^{-1/2} \, (r - r_{+})^{\gamma} \, 
e^{- (\mu^2 - \lambda^2)^{1/2} \, r} 
M(\alpha,2 \gamma, 
2 (\mu^2 - \lambda^2)^{1/2} \, (r-r_{+})) \, \, ,
\end{equation*} 
is exponentially growing for $r \rightarrow \infty$
if $\alpha \notin - {\mathbb{N}}$. Hence, by use of 
Lemma~\ref{eigenfunctionsI}, we arrive at
 
\begin{theorem} \label{basictheorem}
If $\lambda \in {\mathbb{C}} \setminus 
((-\infty,-\mu] \cup [\mu,\infty))$, and the solutions of 
(\ref{eqgamma}) are not part of $(1/2).{\mathbb{Z}}$, then 
$\ker(A - \lambda B - \lambda^2)$ is non-trivial 
if and only if $\lambda \in {\mathbb{C}} \setminus S$
and 
\begin{align} \label{spectralequation}
\frac{1}{2} & + \frac{1}{R_{-}}
\left[- 4 M^2 (2 M R - a^2) \lambda^2 - 
4 m a M R \lambda - 
m^2 a^2 - (M^2 - a^2) + \frac{R_{-}^2}{4} \right]^{1/2} \nonumber 
\\ 
& + \frac{2M (R + M)}{R_{-}} \, (\mu^2 - \lambda^2)^{1/2} + 
\frac{
c_l 
- 2 M (R + 2 M) \mu^2}{2R_{-}(\mu^2 - \lambda^2)^{1/2}} 
= - n 
\end{align}
for some $n \in {\mathbb{N}}$. In the latter case,
$\ker(A - \lambda B - \lambda^2)$ consists of the 
multiples of 
\begin{equation} \label{eq:eigenfunction} 
\triangle^{-1/2} \, (r - r_{+})^{\gamma} \, 
e^{- (\mu^2 - \lambda^2)^{1/2} \, r} 
M(-n,2 \gamma, 
2 (\mu^2 - \lambda^2)^{1/2} \, (r-r_{+})) \, \, ,
\end{equation} 
where 
$\gamma$ is given by (\ref{L2gamma}).
\end{theorem}

Since in the case that $\mu^2 \leq 
[2M(R+2M)]^{-1} c_l$, 
the left hand side of (\ref{spectralequation}) 
has real part $> 1/2$, whereas the right hand side 
has real part $\leq 0$, we obtain the following 
corollary.

\begin{corollary} \label{lowerbound}
If 
\begin{equation*}
 \mu^2 \leq \frac{l(l+1)}{2 M (R + 2 M)} 
\end{equation*}
$\lambda \in {\mathbb{C}} \setminus {\mathbb{R}}$, 
and the solutions of 
(\ref{eqgamma}) are not part of $(1/2).{\mathbb{Z}}$, then 
$\ker(A - \lambda B - \lambda^2)$ is trivial. 
\end{corollary}

In particular, in the case of vanishing mass $\mu = 0$ and for 
the considered cases, 
this implies that 
$\ker(A - \lambda B - \lambda^2)$ is trivial.
Also, we obtain the following estimate for the case that 
the previous inequality is violated.

\begin{corollary}
If 
\begin{equation*}
 \mu^2 > \frac{l(l+1)}{2 M (R + 2 M)} 
\end{equation*}
$\lambda \in {\mathbb{C}} \setminus {\mathbb{R}}$, 
the solutions of 
(\ref{eqgamma}) are not part of $(1/2).{\mathbb{Z}}$ and  
$\ker(A - \lambda B - \lambda^2)$ is non-trivial, then 
\begin{equation} \label{spectralestimate}
|\lambda| \leq 
\left\{\mu^2 + \frac{[2M(R+2M)\mu^2 
- l(l+1)]^2}{R_{-}^2} \right\}^{1/2} \, \, .
\end{equation}
\end{corollary}

\begin{proof}
Straightforward estimates yield that the inequality 
\begin{equation*}
|\lambda^2 - \mu^2| > \frac{[2M(R+2M)\mu^2 
- l(l+1)]^2}{R_{-}^2}
\end{equation*}
implies that 
\begin{equation*}
\frac{1}{2} - {\mathrm Re} \left[ \frac{
2 M (R + 2 M) \mu^2 - c_l}{2R_{-}(\mu^2 - \lambda^2)^{1/2}}
\right]
> 0 
\end{equation*}

and hence that 
the left hand side of (\ref{spectralequation}) 
has real part $> 0$, whereas the right hand side 
has real part $\leq 0$, which contradicts that 
$\ker(A - \lambda B - \lambda^2)$ is non-trivial.$\lightning$
Hence, 
\begin{equation*}
|\lambda^2 - \mu^2| \leq \frac{[2M(R+2M)\mu^2 
- l(l+1)]^2}{R_{-}^2} \, \, .
\end{equation*}
The latter inequality implies (\ref{spectralestimate}).
\end{proof}

\begin{lemma}
If $\lambda$ is such that the solutions of 
(\ref{eqgamma}) are not part of $(1/2).{\mathbb{Z}}$,
then $\lambda \in {\mathbb{R}} \times (-\infty,0)$ is such 
that   
$\ker(A - \lambda B - \lambda^2)$ is non-trivial if and only if  
\begin{equation*}
\lambda = - \, \frac{i}{2z} \, (z^2 - \mu^2)
\end{equation*}
for $z \in ({\mathbb{C}} \setminus B_{\mu}(0)) 
\cap ((0,\infty) \times {\mathbb{R}})$ satisfying 
\begin{align*}
& M^2 (M^2 + a^2 + R^2) z^8 +
 2 M \left[\left(n + \frac{1}{2}\right) R_{-} (R + M) -
    i m a R\right] z^7 \\
& + \left[ (M^2 - a^2) + 
    m^2 a^2 + n(n+1) R_{-}^2 + 
    2 M (R + M) (c_{l} - 2 M^2 \mu^2) \right] z^6 \\
&
+ 2 \left\{ 
\left(n + \frac{1}{2}\right) R_{-} \left[c_{l} + M (R - M) \mu^2 \right] 
-
    i m a M R \mu^2  \right\} z^5 \\
& + \bigg\{c_{l}^2 + 
2 \left[ (M^2 - 
          a^2) + m^2 a^2 + n(n+1) R_{-}^2 - 
2 M^2 c_{l}\right] \mu^2  \\
& \quad \, \, \, \, \, + 
    2 M^2 (R^2 + 4 M R + 3 M^2 - a^2) \mu^4 \bigg\} 
z^4 \\
& + 
 2 \left\{ \left(n + \frac{1}{2}\right) 
R_{-}
\left[c_{l} + 
       M (R - M) \mu^2 \right] + 
i m a M R \mu^2 \right\} \mu^2 z^3 \\
& + \left[ (M^2 - a^2) + m^2 a^2 + n(n+1) R_{-}^2 
    + 2 M (R + M) (c_{l} - 2 \mu^2  M^2 ) \right] \mu^4 z^2 \\
& + 
 2 M \left[\left(n + \frac{1}{2}\right) R_{-} (R + M) + 
    i m a R\right]  \mu^6 z  + M^2 (M^2 + a^2 + R^2) \mu^8 = 0
\end{align*}
and such that
\begin{equation*}
- \left(n + \frac{1}{2}\right) R_{-} - M(R+M) \, \frac{z^2 + \mu^2}{z}
- [c_l - 2M (R+2M) \mu^2] \, \frac{z}{z^2 + \mu^2} 
\end{equation*}
is part of $(0,\infty) \times {\mathbb{R}}$. 
\end{lemma} 

\begin{proof}
The statement follows by direct calculation from 
Theorem~\ref{basictheorem} and Lemma~\ref{conformalmap}.
\end{proof}

\begin{lemma}
If $R = r_{+}$, i.e., $R_{-} = 2 (M^2 - a^2)^{1/2}$, 
and $\lambda$ satisfies (\ref{exceptions1}),
then $\lambda \in {\mathbb{R}} \times (-\infty,0)$ is such 
that   
$\ker(A - \lambda B - \lambda^2)$ is non-trivial if and only if  
\begin{equation*}
\lambda = - \, \frac{i}{2z} \, (z^2 - \mu^2)
\end{equation*}
for $z \in ({\mathbb{C}} \setminus B_{\mu}(0)) 
\cap ((0,\infty) \times {\mathbb{R}})$ satisfies 
\begin{align}
z^4 & + \frac{(2n + 1)(M^2 - a^2)^{1/2} + i ma}{M(2r_{+}+M)}
\, z^3 + \frac{c_{l}- 2 M^2 \mu^2}{M(2r_{+}+M)} \, z^2  \nonumber \\
& + \frac{(2n + 1)(M^2 - a^2)^{1/2} + i ma}{M(2r_{+}+M)} 
\, \mu^2 z + \frac{M}{2r_{+}+M} \, \mu^4 = 0 \, \, .  \label{spectralcondition}
\end{align}

\end{lemma}

\begin{proof}
The statement follows by direct calculation from 
Theorem~\ref{basictheorem} and Lemma~\ref{conformalmap}.
\end{proof}

We note that if $z \in ({\mathbb{C}} \setminus B_{\mu}(0)) 
\cap ((0,\infty) \times {\mathbb{R}})$ satisfies 
(\ref{spectralcondition}), then $z^{*}$ is a solution 
of (\ref{spectralcondition}), where $m$ is replaced by 
$-m$, that is contained in 
$({\mathbb{C}} \setminus B_{\mu}(0)) 
\cap ((0,\infty) \times {\mathbb{R}})$.

\section{Unstable modes}
In this section we proceed to solve~\eqref{spectralcondition} and show
the existence of unstable modes in a subregion of the parameter
space. Although the roots of any forth degree polynomial, such as that
in~\eqref{spectralcondition}, are known explicitly, the 
 expressions, seen as functions of all the parameters involved
(i.e. $a$, $\mu$, $n$, $l$ and $m$)\footnote{Since the solutions are scalable
  in $M$ we will only present results for $M=1$.}, are too complicated to get
any
intuitive understanding of the problem just from their analytical form. Thus, we
find it more insightful to simply do a direct numerical evaluation at different
parameters' values. This approach allows us to find the unstable modes
described below.
\newline
\linebreak
Before continuing, we mention a few words about the numerical errors. 
Besides conducting error estimations of the eigenvalues found, we also checked
our results by finding the roots with the  Newton Raphson (complex) method and 
then compare them with the values found by
direct evaluation, giving identical results up to error. 
We find that the errors  can be made as small as desired just by
increasing the number of figures used in the calculations.
All values presented in this paper are accurate at least up to the last
figure shown.
\newline
\linebreak
We find unstable modes in a range of values of
$\mu$, which increases with $l$, $m$ and $a$. That region decreases with $n$, but the
dependence on $n$ seems to be much smaller than that on  $l$, $m$ and $a$, and
hence for the majority of the results presented 
here we just set $n=0$. Let the reader assume $n=0$ from now on unless otherwise specified.
We find unstable modes for $l$ and $m$ starting at 3 and large $a$, mostly
with $l=m$, although for large $l$ we also find unstable modes for some $m<l$. To
give a few examples, for $l=m=3$ we find instabilities with $a=0.9976$ and larger,
while for $l=m=10$ we find them for $a=0.9833$ and larger. A more complete
description is given in the figures. 
The curves in Figure~\ref{lambda_al}
represent unstable regions in $\mu$ corresponding to four values of $a$, for
fixed $l=m$; and to four values of $l=m$ and fixed $a$.
\begin{figure} 
\centering
\includegraphics[width=0.50\textwidth,height=!]{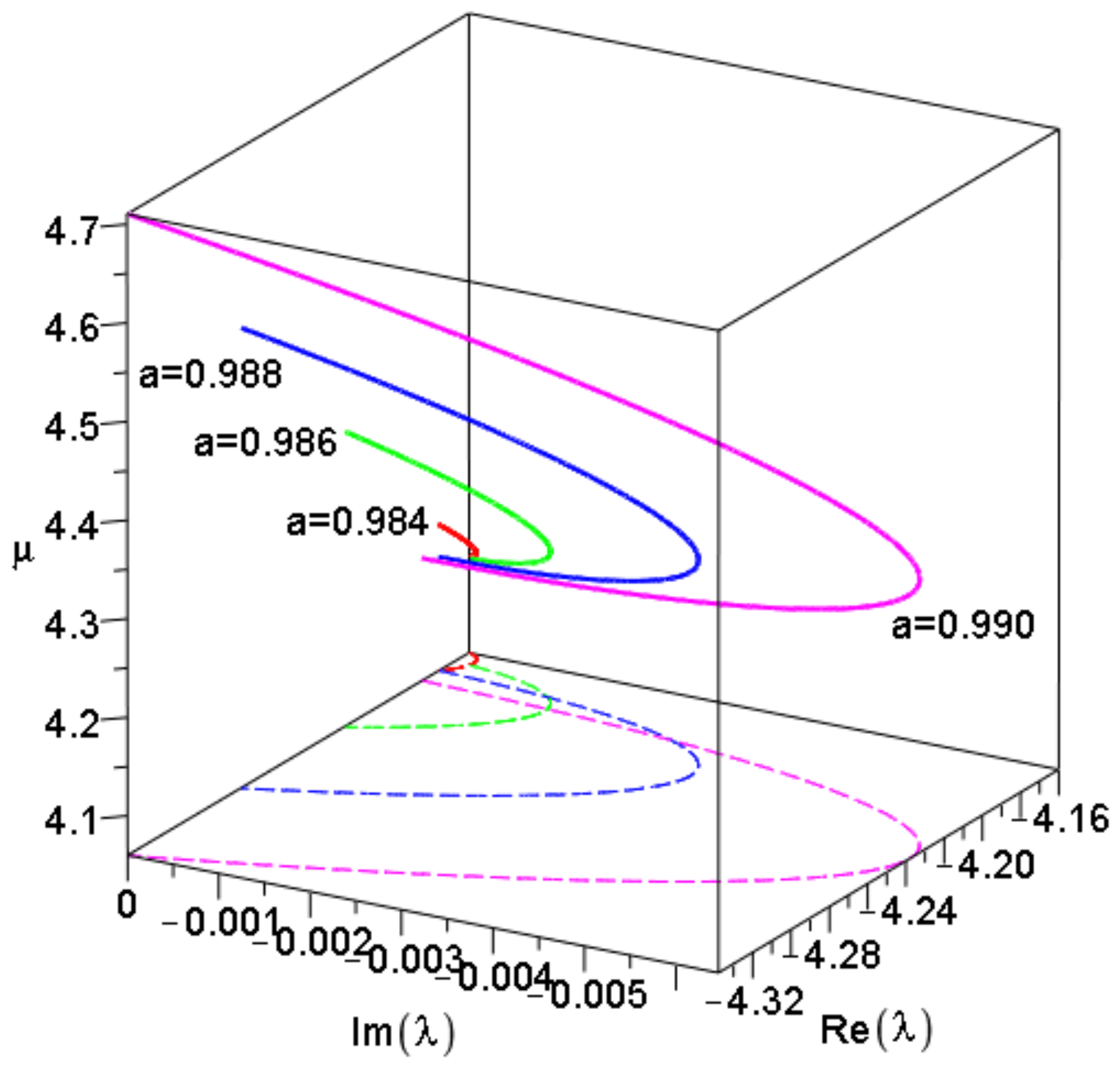}
\includegraphics[width=0.49\textwidth,height=!]{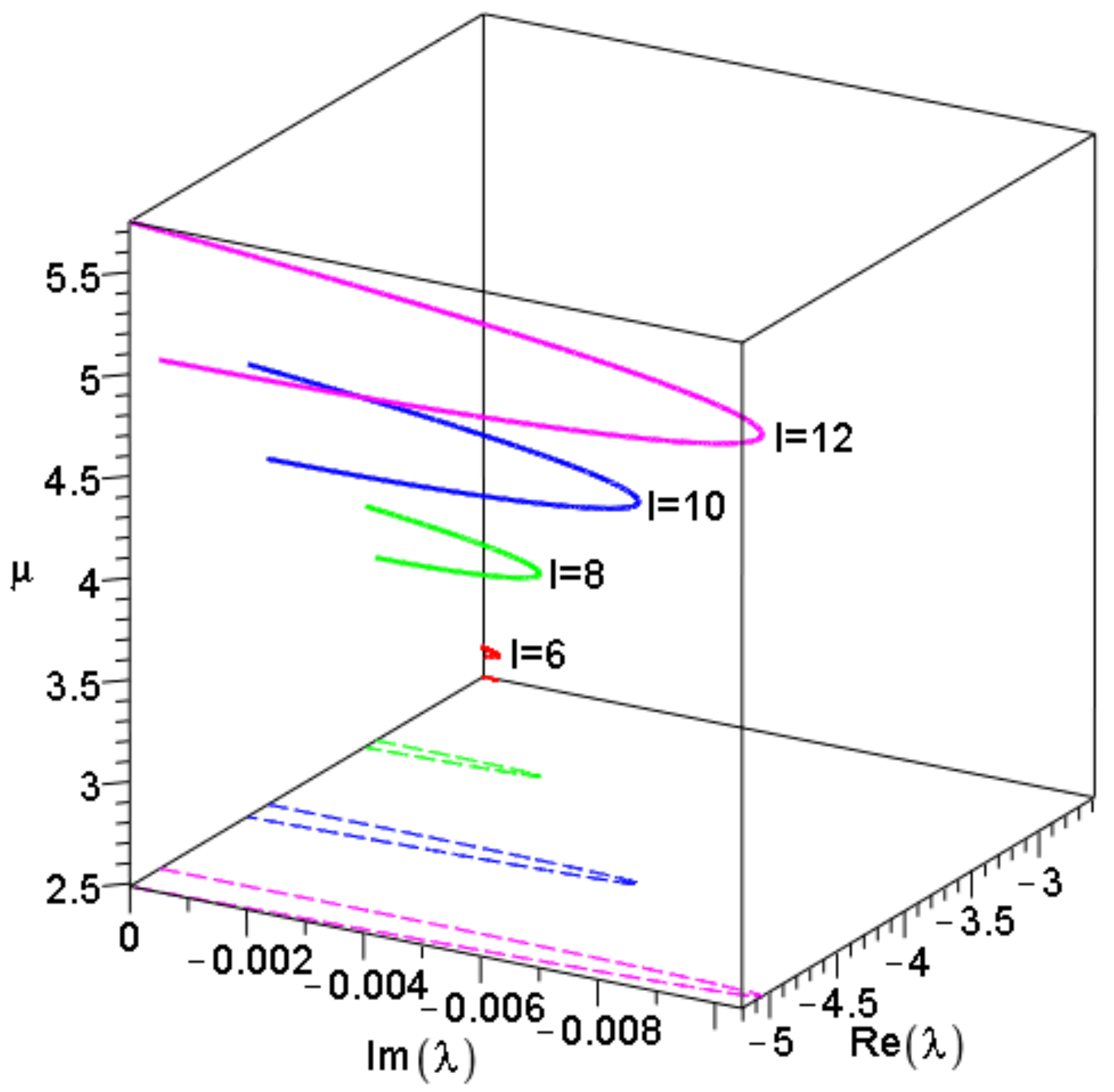}
\caption{Eigenvalues $\lambda$ of unstable modes ($\lambda<0$) for varying $\mu$, $a$ and $l=m$. 
$\mu$ is shown in the vertical axis, while the horizontal axes represent the real and imaginary part of
$\lambda$. The dashed lines are a projection of the curves into the complex plane to improve visibility.
{\em left panel:} $l=m=10$; $a=0.984$ (red), 0.986 (green), 0.988 (blue) and
0.99 (magenta). 
{\em right panel:} $a=0.99$; $l=m=6$ (red), 8 (green), 10 (blue) and 12 (magenta).}
\label{lambda_al}
\end{figure}
In each case, we show the complete range of values $\mu$ for which we found
unstable modes (vertical axis of the figures).
Interestingly, in all the cases studied (including others not shown in the
figures) that range starts exactly (up to numerical accuracy) at the 
limit given in Corollary~\ref{lowerbound}. Finally, in
Figure~\ref{f:eigenfunctions} we show the norm of the
eigenfunctions~\eqref{eq:eigenfunction} (see Theorem~\ref{t:eigenfunction}) as a
function of $r-R_{+}$ for some values of $\mu$. In the inset we use a logarithmic scale in
the horizontal axis to show the details close to $R_{+}$.
\begin{figure} 
\centering
\includegraphics[width=0.7\textwidth,height=!]{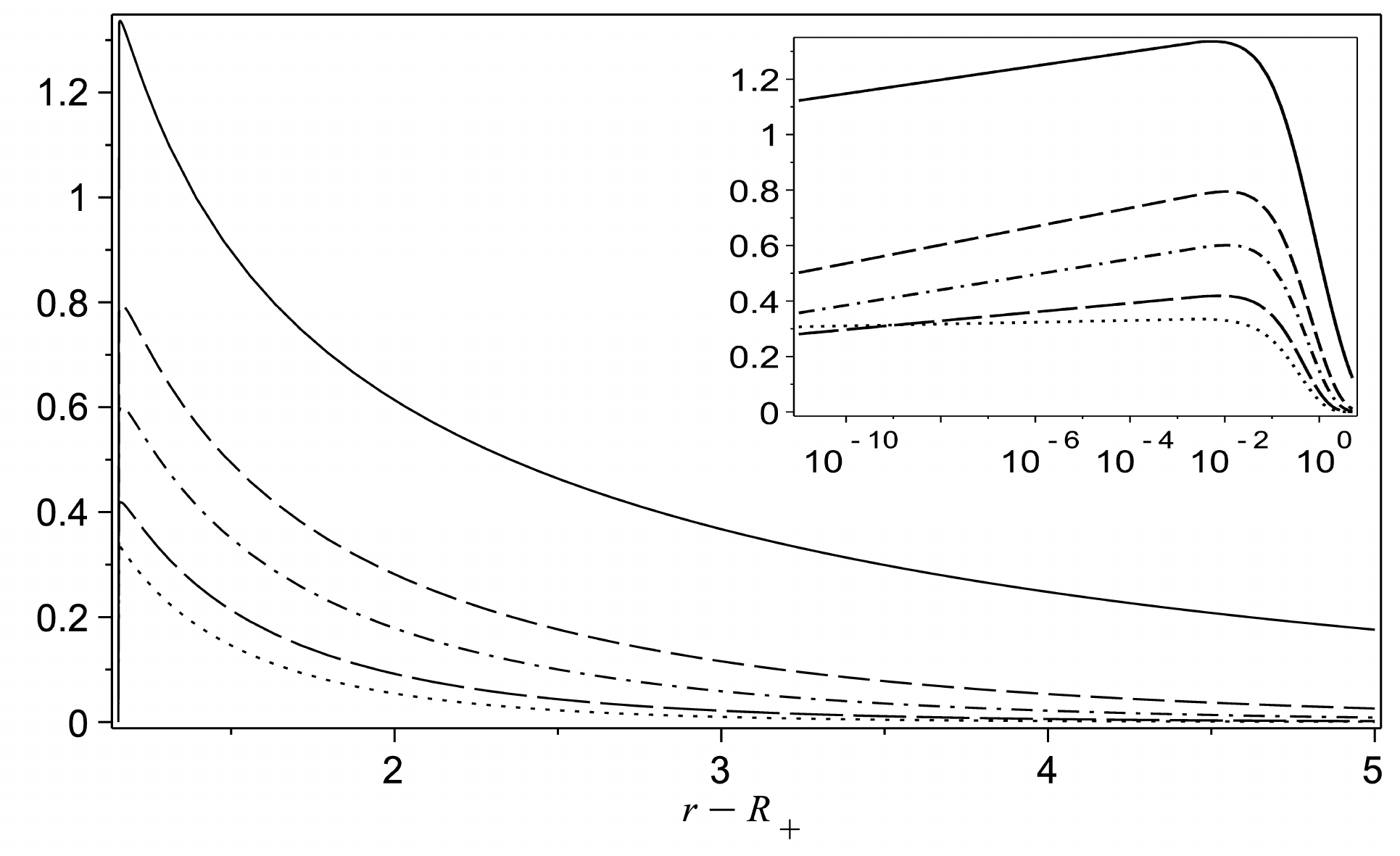}
\caption{Eigenfunctions of unstable modes with $n=0$, $l=m=10$, $a=0.988$; and
  $\mu=4.184$ (solid), $4.237$ (dash), $4.290$ (dot dash), $ 4.396$ (long dash), $4.502$ 
(dots).}
\label{f:eigenfunctions}
\end{figure}
\newline
\linebreak
Since previous works have reported instabilities only for large values of $a$,
it is worth searching for instabilities with the smallest possible $a$.
Figure~\ref{aminl} shows the minimum value of $a$, $a_{\rm min}$, for which we could find
instabilities for given $l$. For this task we only investigated modes with $m=l$, since
usually setting  $m<l$ gives rise to less (if any) unstable modes. 
Although an exhaustive search was made to determine $a_{\rm min}$, it is
impossible to strictly rule out the possible existence of unstable modes
beyond $a_{\rm min}$
with our method of direct evaluation, since an infinite number of such
evaluations would be necessary to completely cover any subregion of the
parameter space (unless the continuous parameters, $a$ and $\mu$, are both
fixed).
In general, for given $a$ a range of $\mu$ values give unstable
modes. If $a$ is then decreased, that range gets smaller until, at $a=a_{\rm
  min}$, only one value of $\mu$ gives unstable modes (see the left panel of
figure~\ref{lambda_al}). Each of these values are show as annotations in Figure~\ref{aminl}.
\begin{figure} 
\centering
\includegraphics[width=0.7\textwidth,height=!]{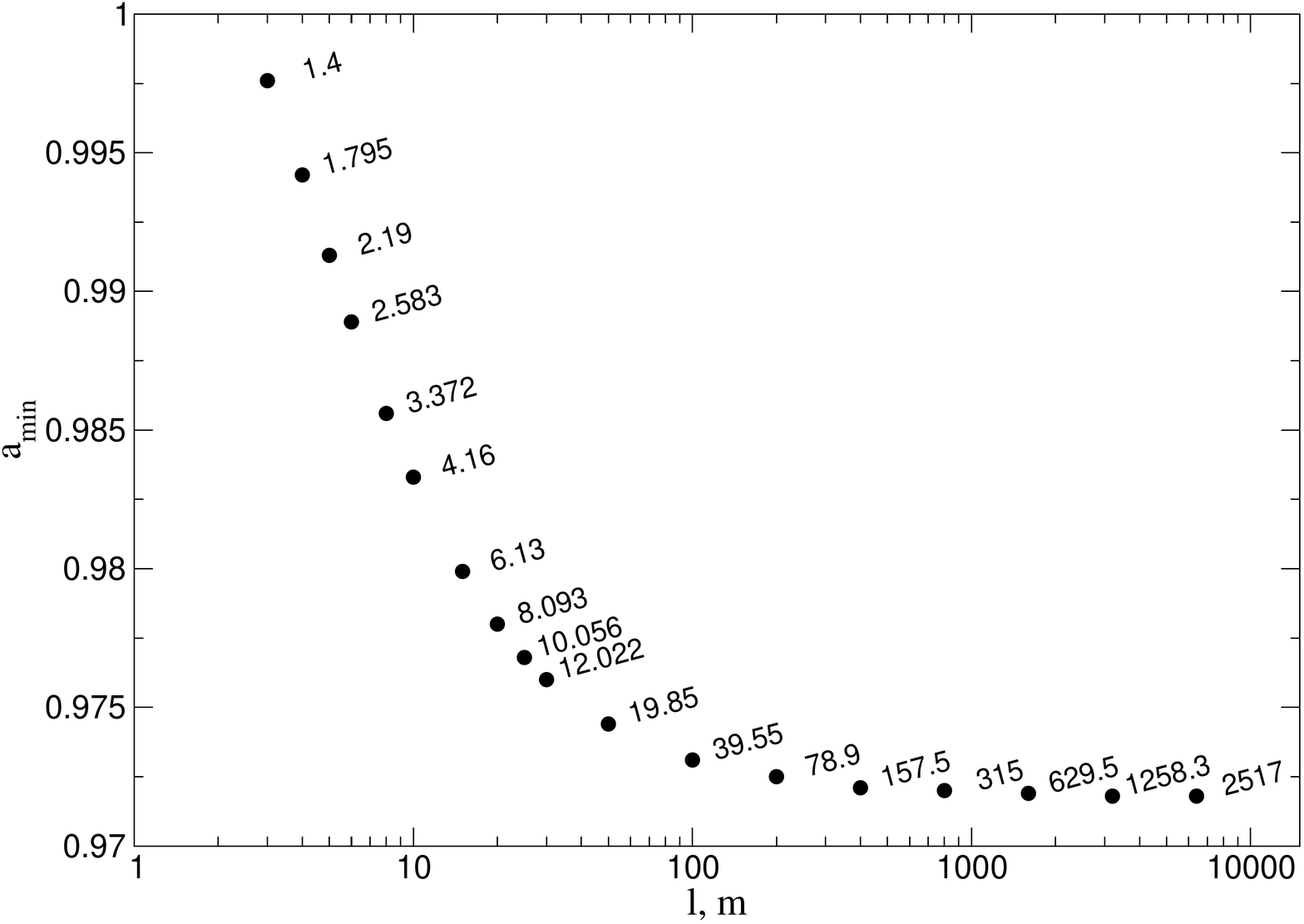}
\caption{Minimum $a$ for which unstable modes were found for given $l=m$. The
  annotated numbers show the approximate value of $\mu$.}
\label{aminl}
\end{figure}
As seen in the figure, $a_{\rm min}$ decreases with $l$ and seems to asymptote
to a value $a_{\rm asym.}\approx0.9718$. It would be interesting to determine the mathematical root of this
threshold. This task will be pursued in a future work.


\section{Discussion of the Results}

From the Klein-Gordon equation in a Kerr background, we 
 derive a simplified model equation showing the 
instability of the reduced, by separation in the azimuth 
angle in Boyer-Lindquist coordinates, field for a certain range 
of the masses $\mu$ of the field. In addition, we give 
a well-posed initial value formulation of that reduced 
equation, along with a stability analysis of the corresponding 
field. The latter shows instability down to rotational parameters 
$a/M \approx 0.97$. This result supports claims 
of previous analytical and numerical investigations that show 
instability of 
the massive Klein-Gordon field for $a$ 
extremely close to $1$. 
\newline
\linebreak 
From here, mathematical investigation could proceed 
in $2$ directions. First, it might be possible, to use the 
model for the proof of the instability of the massive 
Klein-Gordon equation in a Kerr background, using a perturbative 
approach. Another direction consists in further analysis
of the model in order to find the mathematical root 
of the instability as well as an abstraction to a larger 
class of equations that includes the massive Klein-Gordon 
equation on a Kerr background.

\section*{Acknowledgments} 
H.B. is thankful for the hospitality and support
of the `Department of Gravitation and Mathematical Physics',
(ICN, Miguel Alcubierre), 
Universidad Nacional Autonoma de 
Mexico, Mexico City, Mexico and 
the `Division for Theoretical Astrophysics' (TAT,  
K. Kokkotas) of the Institute for Astronomy and 
Astrophysics   
at the Eberhard-Karls-University Tuebingen. 
This work was supported in part by CONACyT
grants 82787 and 167335, DGAPA-UNAM through grant IN115311, SNI-M\'exico, and 
the SFB/Transregio 7 on 
``Gravitational Wave Astronomy'' of the German Science Foundation
(DFG). M.M. acknowledges DGAPA-UNAM for a postdoctoral grant.
 
\section{Appendix}

\begin{figure} 
\centering
\includegraphics[width=5.6cm,height=5.6cm]{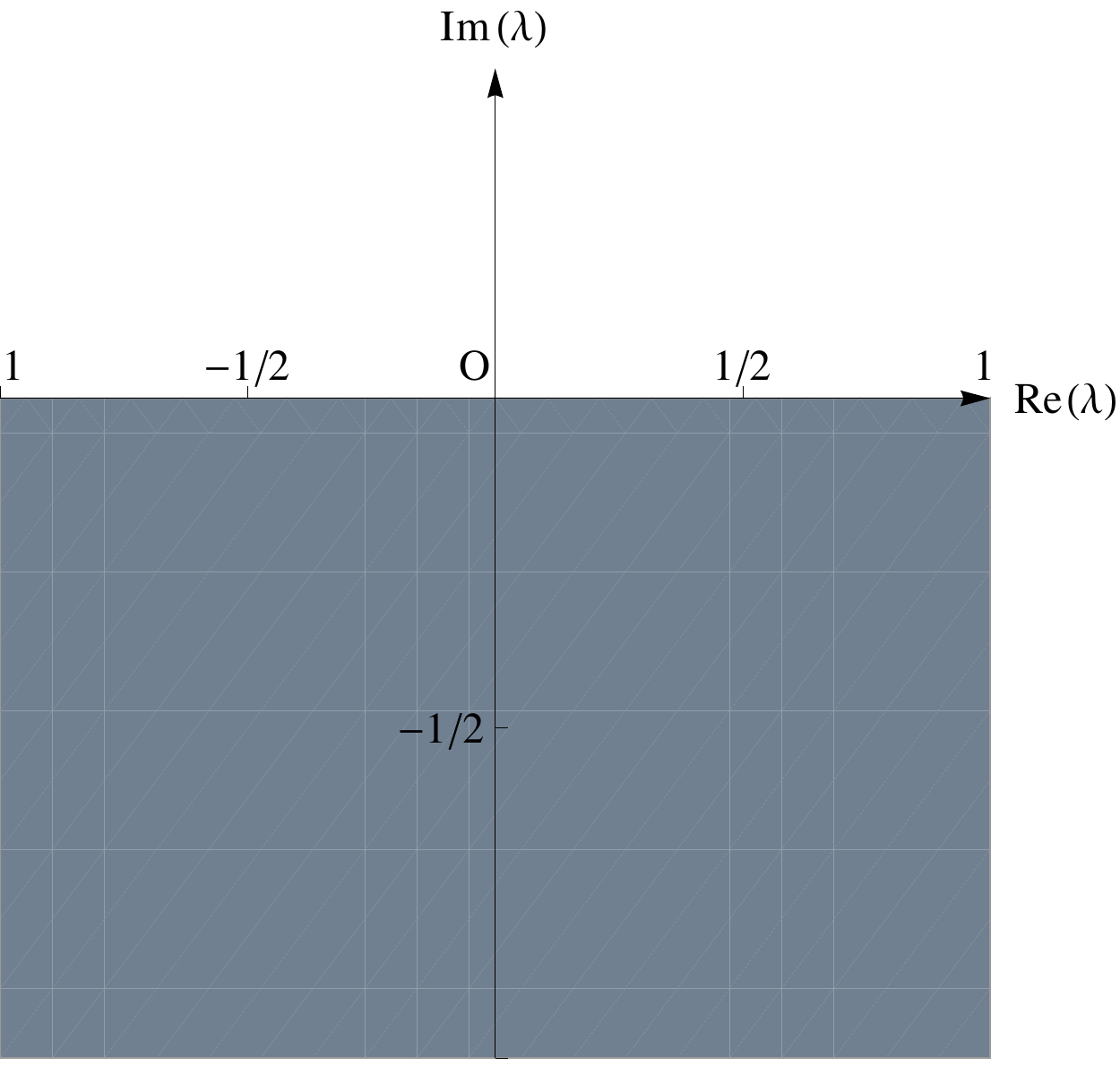}
\includegraphics[width=5.6cm,height=5.6cm]{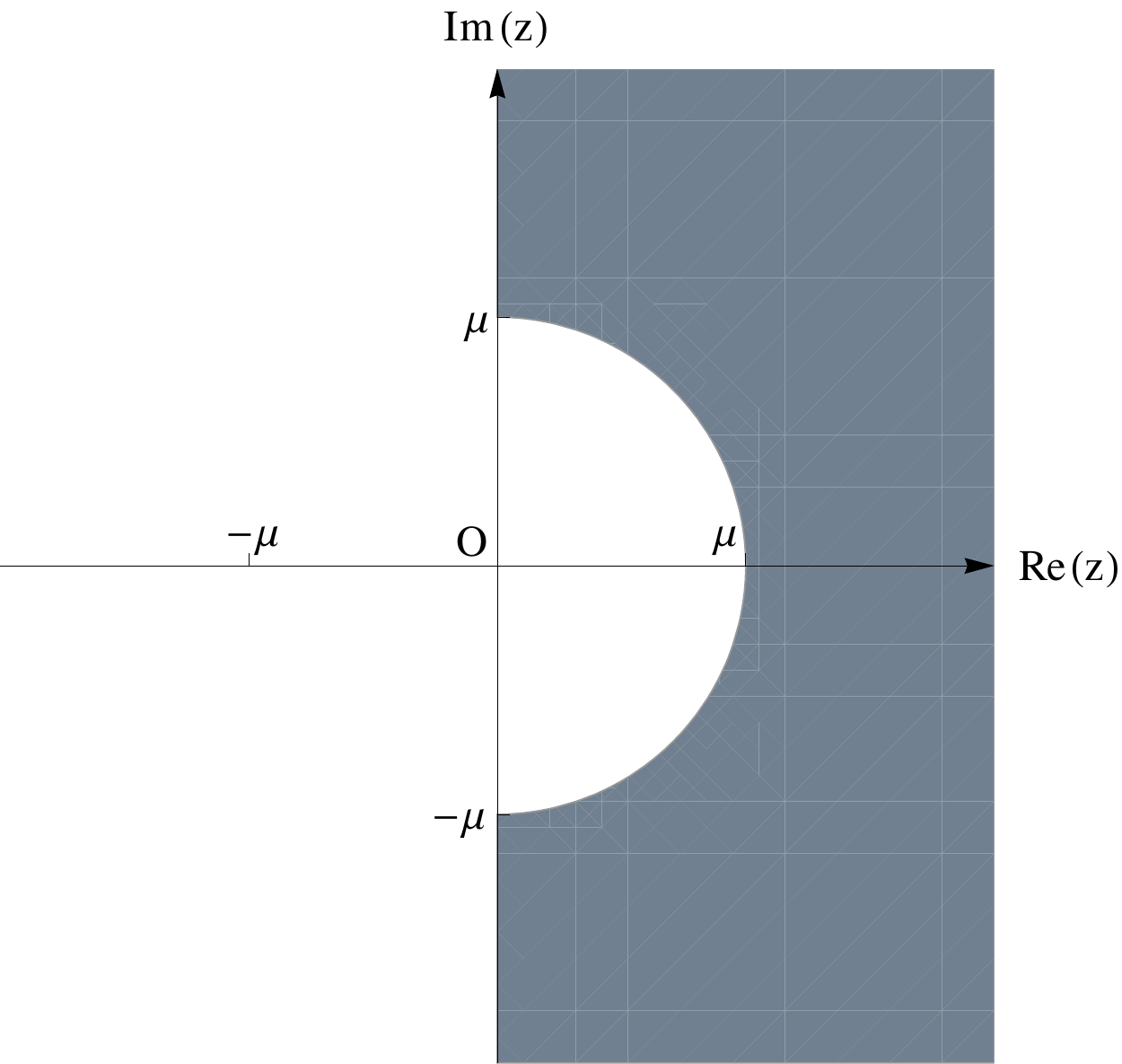}
\caption{Sketch of domain and range of the map $g$ from 
Lemma~\ref{conformalmap}.}
\label{fig1}
\end{figure}

\begin{lemma} \label{conformalmap}
Let $\mu \in [0,\infty)$. By 
\begin{equation*}
g(\lambda) := i \lambda + (\mu^2 - \lambda^2)^{1/2} 
\end{equation*}
for every $\lambda \in {\mathbb{R}} \times (-\infty,0)$, 
where 
$(\mu^2 - \lambda^2)^{1/2}$ denotes the square root with 
strictly positive real part, there is defined a 
biholomorphic map 
\begin{equation*}
g : {\mathbb{R}} \times (-\infty,0) \rightarrow ({\mathbb{C}} \setminus B_{\mu}(0)) 
\cap ((0,\infty) \times {\mathbb{R}}) 
\end{equation*}
with inverse 
\begin{equation*}
g^{-1} : ({\mathbb{C}} \setminus B_{\mu}(0)) 
\cap ((0,\infty) \times {\mathbb{R}}) \rightarrow 
{\mathbb{R}} \times (-\infty,0) 
\end{equation*}
given by 
\begin{equation*}
g^{-1}(z) = - \, \frac{i}{2z} \, (z^2 - \mu^2) 
\end{equation*}
for every $z \in ({\mathbb{C}} \setminus B_{\mu}(0)) 
\cap ((0,\infty) \times {\mathbb{R}})$. In addition, 
\begin{equation*}
\left[\mu^2 - (g^{-1}(z))^2 \right]^{1/2} = \frac{1}{2 z} \, (z^2 + \mu^2) 
\end{equation*}
for every $z \in ({\mathbb{C}} \setminus B_{\mu}(0)) 
\cap ((0,\infty) \times {\mathbb{R}})$.
\end{lemma}

\begin{proof}
First, we note that $g$ is well-defined. For this, 
let $\lambda \in 
{\mathbb{R}} \times (-\infty,0)$ and 
$\lambda_1 := \textrm{Re}(\lambda),
\lambda_2 := \textrm{Im}(\lambda) < 0$. Then
\begin{equation*}
\mu^2 - \lambda^2 = \mu^2 - \lambda_1^2 + \lambda_2^2 - 2 i \lambda_1 \lambda_2
\, \, .
\end{equation*}
Hence $\mu^2 - \lambda^2$ is real iff $\lambda_1 = 0$. In the latter case,
\begin{equation*}
\textrm{Re}\,(\mu^2 - \lambda^2) = \mu^2 + \lambda_2^2 > 0 \, \, .
\end{equation*}
Therefore, $\mu^2 - \lambda^2 \in {\mathbb{C}} \setminus ((-\infty,0] \times \{0\})$, 
and there is precisely one square root 
of $\mu^2 - \lambda^2$ with strictly positive
real part.  
Further, 
\begin{equation*}
\textrm{Re}\,[\,i \lambda +(\mu^2 - \lambda^2)^{1/2}\,] 
= - \lambda_2 
+ \textrm{Re} \, (\mu^2 - \lambda^2)^{1/2} > 0 \, \, . 
\end{equation*}
In particular, if 
\begin{equation*}
z := i \lambda + (\mu^2 - \lambda^2)^{1/2}
\end{equation*}
is such that $|z| \leq \mu$, then
\begin{equation*}
z^2 - 2 i z \lambda - \lambda^2  = (z - i \lambda)^2 = \mu^2 - \lambda^2 \, \, .
\end{equation*}
Hence 
\begin{equation*}
z^2 - \mu^2 = 2 i z \lambda 
\end{equation*}
and 
\begin{equation*}
\lambda = - \frac{i}{2} \left(z - \frac{\mu^2}{z}\right) \, \, . 
\end{equation*}
The latter implies that 
\begin{equation*}
\lambda_2 = \textrm{Im} \lambda = - \frac{\textrm{Re}(z)}{2} \left( 
1 - \frac{\mu^2}{|z|^2}
\right) 
\end{equation*}
and hence that $\textrm{Re}(z) \leq 0$.$\lightning$
As a consequence, $|z| > \mu$. 
\newline
\linebreak 
For the second step, let $z \in ({\mathbb{C}} \setminus B_{\mu}(0)) 
\cap ((0,\infty) \times {\mathbb{R}})$, $x := \textrm{Re}(z),$ and 
$y := \textrm{Im}(z)$. Then 
\begin{align*}
& - \frac{i}{2} \left(z - \frac{\mu^2}{z}\right) = 
- \frac{i}{2} \left(x + iy - \frac{\mu^2}{x+iy} \right) =
- \frac{i}{2} \left[x + iy - \frac{\mu^2}{|z|^2} \, (x-iy) \right] \\
& = 
\frac{1}{2} \left[y 
\left(1 + \frac{\mu^2}{|z|^2}\right) - i x 
\left(1 - \frac{\mu^2}{|z|^2}\right) \right] \in {\mathbb{R}} \times (-\infty,0)
\end{align*}
and 
\begin{align*}
g\left(- \frac{i}{2} \left(z - \frac{\mu^2}{z}\right) \right) & = 
\frac{1}{2} \left(z - \frac{\mu^2}{z}\right) + \left[\mu^2 + 
\frac{1}{4} \left(z - \frac{\mu^2}{z}\right)^2\right]^{1/2} \\
& =
\frac{1}{2} \left(z - \frac{\mu^2}{z}\right) + \left[ 
\frac{1}{4} \left(z + \frac{\mu^2}{z}\right)^2 \right]^{1/2} 
\, \, .
\end{align*}
Since 
\begin{align*}
& \frac{1}{2} \left(z + \frac{\mu^2}{z}\right) = 
\frac{1}{2} \left[x + iy + \frac{\mu^2}{|z|^2} \, (x-iy) \right] \\
& =
\frac{1}{2} \left[x \left(1 + \frac{\mu^2}{|z|^2}\right) 
+ iy \left(1 - \frac{\mu^2}{|z|^2}\right)\right] \in (0,\infty) \times {\mathbb{R}}
\, \, , 
\end{align*}
the latter implies that 
\begin{align*}
& g\left(- \frac{i}{2} 
\left(z - \frac{\mu^2}{z}\right) \right) = 
\frac{1}{2} \left(z - \frac{\mu^2}{z}\right) + \left[ 
\frac{1}{4} \left(z + \frac{\mu^2}{z}\right)^2 \right]^{1/2} \\
& =  
\frac{1}{2} \left(z - \frac{\mu^2}{z}\right) + 
\frac{1}{2} \left(z + \frac{\mu^2}{z}\right) = z \, \, .
\end{align*}
In particular, we conclude that by 
\begin{equation*}
h(z) := - \, \frac{i}{2} \left(z - \frac{\mu^2}{z} \right) = - 
\frac{i}{2z}\, (z^2 - \mu^2)
\end{equation*}
for every $z \in ({\mathbb{C}} \setminus B_{\mu}(0)) 
\cap ((0,\infty) \times {\mathbb{R}})$, there is defined a map 
\begin{equation*}
h : ({\mathbb{C}} \setminus B_{\mu}(0)) 
\cap ((0,\infty) \times {\mathbb{R}}) \rightarrow 
{\mathbb{R}} \times (-\infty,0) 
\end{equation*}
such that 
\begin{equation*}
(g \circ h)(z) = z 
\end{equation*}
for every $z \in ({\mathbb{C}} \setminus B_{\mu}(0)) 
\cap ((0,\infty) \times {\mathbb{R}})$. Therefore,  
$g$ is surjective. 
\newline
\linebreak
Further, if $\mu = 0$,
\begin{equation*}
(h \circ g)(\lambda) = - \, \frac{i}{2} \, [\,
i \lambda + (- \lambda^2)^{1/2}\,] = - \, \frac{i}{2} \, (\,
i \lambda + i \lambda \, ) = \lambda 
\end{equation*}
and if $\mu \neq 0$,
\begin{align*}
& (h \circ g)(\lambda) = - \, \frac{i}{2} \left[
i \lambda + (\mu^2 - \lambda^2)^{1/2} 
- \frac{\mu^2}{i \lambda + (\mu^2 - \lambda^2)^{1/2}} 
\right] \\
& =
- \, \frac{i}{2} \left[i \lambda + (\mu^2 - \lambda^2)^{1/2} 
- \mu^2 \, \frac{i \lambda - (\mu^2 - \lambda^2)^{1/2}}
{-\lambda^2 - (\mu^2 - \lambda^2)} \right] \\
& = 
- \, \frac{i}{2} \left[i \lambda + (\mu^2 - \lambda^2)^{1/2} 
+ 
i \lambda - (\mu^2 - \lambda^2)^{1/2}\right] = \lambda 
\end{align*}
for every $\lambda \in {\mathbb{R}} \times (-\infty,0)$. Therefore, 
$g$ is injective and altogether bijective with inverse $h$. Finally, it follows for 
$\lambda \in {\mathbb{R}} \times (-\infty,0)$ that 
\begin{align*}
& - (\mu^2 - \lambda^2)^{1/2} = i \lambda - g(\lambda) 
=  i h(g(\lambda)) - g(\lambda)
=  - i \, \frac{i}{2} \left(g(\lambda) - \frac{\mu^2}{g(\lambda)} \right) 
- g(\lambda) \\
& = - \frac{1}{2} \left(g(\lambda) + \frac{\mu^2}{g(\lambda)} \right) = 
- \frac{1}{2 g(\lambda)} \, [(g(\lambda))^2 + \mu^2]
\end{align*}
and hence for every $z \in ({\mathbb{C}} \setminus B_{\mu}(0)) 
\cap ((0,\infty) \times {\mathbb{R}})$ that 
\begin{equation*}
\left[\mu^2 - (g^{-1}(z))^2\right]^{1/2} 
= \frac{1}{2 z} \, (z^2 + \mu^2) \,\, .
\end{equation*}
\end{proof}

\pagebreak

\end{document}